\documentclass[letterpaper, 10 pt, conference]{ieeeconf}  

\IEEEoverridecommandlockouts                              
\usepackage{amsmath,amssymb,bm,bbm,mathrsfs}
\usepackage{color, subfig}
\usepackage{graphicx}
\usepackage{epstopdf}
\usepackage{epsfig}
\usepackage{graphics} 
\usepackage{epstopdf}
\usepackage{enumerate}
\usepackage{caption}
\usepackage{amsmath} 
\usepackage{amssymb, dsfont}  
\usepackage{subfig,pgfplots}
\pgfplotsset{ 
  compat=newest, 
  legend style =
  {font=\small \sffamily},
  label style = {font=\small\sffamily},
every tick label/.append style={font=\small}}
\usepackage{tikz}

\usetikzlibrary {positioning}

\def\beq{\begin{equation}}
\def\eeq{\end{equation}}
\newtheorem{theorem}{Theorem}
\newtheorem{definition}[theorem]{Definition}

\newtheorem{lemma}[theorem]{Lemma}

\newtheorem{example}{Example}
\newtheorem{remark}{Remark}
\newcommand{\abs}[1]{\lvert#1\rvert}
\newcommand{\ds}{\displaystyle}

\newcommand{\ba}{\begin{array}}
\newcommand{\ea}{\end{array}}

\newcommand{\be}{\begin{equation}}
\newcommand{\ee}{\end{equation}}
\newcommand{\eps}{\varepsilon}

\newcommand{\mc}{\mathcal}

\newcommand{\ov}{\overline}

\newcommand{\1}{\mathbbm{1}}

\newcommand{\R}{\mathbb{R}}

\newcommand{\de}{\mathrm{d}}

\renewcommand{\span}{\text{span}}

\DeclareMathOperator{\dist}{dist}
\def\1{\mathds{1}}

\def\R{\mathbb{R}}

\def\diag{{\rm diag}\,}

\def\x{{\bf x}}

\def\sgn{\text{sgn}}
 \def \l {2.4}
\def \h {1.8}

\title{\LARGE \bf
On imitation dynamics in potential population games
}

\author{Lorenzo Zino, Giacomo Como, and Fabio Fagnani
\thanks{The authors are with the ``Lagrange'' Department of Mathematical Sciences, 
        Politecnico di Torino, 10129 Torino, Italy
        {\tt\small \{lorenzo.zino, giacomo.como, fabio.fagnani\}@polito.it}}%
\thanks{L. Zino is also with the ``Peano'' Department of Mathematics, Universit{\`a} di Torino,
        10123 Torino, Italy
        {\tt\small lorenzo.zino@unito.it}}%
        \thanks{G. Como is also with the Department of Automatic Control, Lund University, 22100 Lund, Sweden
        {\tt\small giacomo.como@control.lth.se}}%
}

\begin{document}

\maketitle
\thispagestyle{empty}
\pagestyle{empty}

\begin{abstract}

Imitation dynamics for population games are studied and their asymptotic properties analyzed. In the considered class of imitation dynamics ---that encompass the replicator equation as well as other models previously considered in evolutionary biology--- players have no global information about the game structure, and all they know is their own current utility and the one of fellow players contacted through pairwise interactions. For potential population games, global asymptotic stability of the set of Nash equilibria of the sub-game restricted to the support of the initial population configuration is proved. These results strengthen (from local to global asymptotic stability) existing ones and generalize them to a broader class of dynamics. The developed techniques highlight a certain structure of the problem and suggest possible generalizations  from the fully mixed population case to imitation dynamics whereby agents interact on complex communication networks.

\end{abstract}

\section{Introduction}

Imitation dynamics provide a powerful game-theoretic paradigm used to model the evolution of behaviors and strategies in social, economic, and biological systems \cite{Weibull1995, Bjornerstedt1996, Hofbauer2003}. The assumption beyond these models is that individuals interact in a fully mixed population having no global information about the structure of the game they are playing. Players just measure their own current utility and, by contacting other individuals, they get aware of the the action currently played by them and of the associated utility. Then, in order to increase their utility, players may revise their action and adopt the one of the contected fellow players. 

We focus on the asymptotic behavior of such imitation dynamics. Available result in this area can be found in \cite{Nachbar1990, Hofbauer2000, Sandholm2001, Sandholm2010}. In particular, \cite{Sandholm2010} contains a study of local stability and instability for the different kinds of rest points of such dynamics. These results, however, deal only with local stability, therefore one can not conclude global asymptotic stability. Indeed, only for specific dynamics, such as the replicator equation, and for some specific classes of games, a global analysis has been carried on  \cite{Bomze2002, Shamma2005, Fox2012, Cressman2014, Barreiro-Gomez2016}. 

This work contributes to expanding the state of the art on the analysis of the asymptotic behavior of imitation dynamics. For the important class of potential population games, we obtain a global convergence result, Theorem \ref{main theorem},  that is stronger and more general than the results presented in the literature. Another novelty of this work consists in the definition of imitation dynamics, that is more general than the classical one \cite{Sandholm2010}. 

The paper is organized as follows.
Section \ref{sec:imitation-dynamics} is devoted to the introduction of population games and to the definition of the class of imitation dynamics. Both these concepts are presented along with some explanatory examples. Thereafter, the main results on the asymptotic behavior of the imitation dynamics are presented and proved in Section \ref{sec:asymptotics}. Examples of the use of these results will then be presented in Section \ref{sec:examples}. Finally, Section \ref{sec:conclusions} discusses some future research lines.

Before moving to the next section, let us define the following notation: $\delta^{(i)}$ denotes a vector of all zeros but a $1$ in the $i$th position.
We denote the sets of reals and nonnegative reals by $\R$ and $\R_+=\{x\in\R:\, x\ge0\}$, respectively. 

\section{Population Games and Imitation Dynamics}\label{sec:imitation-dynamics}
Throughout the paper we study imitation dynamics in continuous population games. 
In such setting, a continuum of players of total unitary mass choose actions from a finite set $\mc A$ and the reward $r_i(x)$ of all those players playing action $i\in\mc A$ is a function of the empirical distribution $x$ of the actions played across the population. 
Formally, let 
$\mc X=\left\{x\in\R_+^{\mc A}:\,\sum\nolimits_{i\in\mc A}x_i=1\right\}$
be the unitary simplex over the action set $\mc A$ and refer to vectors $x\in\mc X$ as \emph{configurations} of the population. If the population is in configuration $x\in\mc X$, then a fraction $x_i$ of the players is playing action $i$, for $i\in\mc A$. 
Let  $r:\mc X\to\R^{\mc A}$ be \emph{reward vector function} whose entries $r_i(x)$ represent the reward received by any player playing action $i\in\mc A$ when the population is in configuration $x\in\mc X$. Throughout, we assume the reward vector function $r(x)$ to be Lipschitz-continuous over the configuration space $\mc X$. Let 
$$r_*(x):=\max_{i\in\mc A}r_i(x)\,,\qquad \ov r(x):=\sum_{i\in\mc A}x_ir_i(x)$$
stand for the maximum and, respectively, the average rewards in a configuration $x\in\mc X$.  
Then, the set of Nash equilibria of the considered continuous population game is denoted by
\be\label{eq:nash}\mc N=\left\{x\in\mc X:\,x_i>0\Rightarrow r_i(x)=r^*(x)\right\}\,.\ee
As is known, every continuous population game admits a Nash equilibrium \cite[Theorem 2.1.1]{Sandholm2010}, so $\mc N$ is never empty.

\begin{example}[Linear reward population games]\label{ex:linear-reward}
A class of continuous population games is the one where the rewards are linear functions of the configuration, i.e.,  when 
\be\label{eq:linear-reward}r(x)=Rx\,,\ee
for some reward matrix $R\in\mathbb R^{\mc A\times\mc A}$. 
Linear reward population games have a standard interpretation in terms of symmetric $2$-player games \cite{Weibull1995} played by each player against the average population \cite{Hofbauer1998}. Population games with binary action space $\mc A=\{1,2\}$ and linear reward function \eqref{eq:linear-reward} with
\be\label{eq:R-binary}
R=\left[\ba{ll}a &b \\ c&d\ea\right]\ee
can be grouped in the following three classes: 
\begin{enumerate}
\item[(i)] for $a>c$ and $d>b$, one has binary coordination games \cite{Cooper1999, Easley2010} (such as the stag hunt game \cite{Skyrms2004}), where the set of Nash equilibria $\mc N=\{\delta^{(1)},\delta^{(2)},\ov x\}$ comprises the two pure configurations and the interior point $\bar x$ with $$\bar x_1=1-\bar x_2=(d-b)/(a-c+d-b)\,;$$ 
\item[(ii)] for $a<c$ and $d<b$, one has anti-coordination games (including hawk-dove game \cite{Rapoport1966, Sugden1986}), where the only Nash equilibrium is the interior point $\bar x$ as above; 
\item[(iii)]
for other cases of the parameters (e.g., in the Prisoner's dilemma \cite{Easley2010}), there is one of the two actions $i$ that is (possibly weakly) dominating the other one $j$, and the pure configuration $\delta^{(i)}$ is the only Nash equilibrium.
\end{enumerate}
Larger action spaces admit no as simple classifications. 
\end{example}\medskip

In this paper, we are concerned with imitation dynamics arising when players in the population modify their actions in response to pairwise interactions  \cite{Hofbauer2003}. We assume that the population is fully mixed so that any pairs of players in the population meet with the same frequency \cite{Kurtz1981}. Upon a possible renormalization, the overall frequency of pairwise interactions between agents playing actions $i$ and $j$ can then be assumed equal to the product $x_ix_j$ of the fraction of players currently playing actions $i$ and $j$, respectively. 
When two players meet, they communicate to each other the action they are playing and the rewards they are respectively getting. Then, depending on the difference between the two rewards and possibly other factors, each interacting player either keeps playing the same action he/she is playing, or updates his/her action to the one of the other player. 

\begin{definition}[Imitation dynamics]
A (deterministic, continuous-time) imitation dynamics for a continuous population game with action set $\mc A$ and reward function vector $r(x)$ is the system of ordinary differential equations 
\be\label{eq:imitation-dynamics}
\dot x_i=x_i\sum_{j\in\mc A} x_j\left(f_{ji}(x)-f_{ij}(x)\right)\quad  i\in\mc A\,,
\ee
where,  for $i,j\in\mc A$, the function $f_{ij}(x)$ is Lipschitz-continuous on the configuration space $\mc X$ and such that 
\be\label{assumption:fij}\sgn\left(f_{ij}(x)-f_{ji}(x)\right)=\sgn\left(r_j(x)-r_i(x)\right),\;\; x\in\mc X.\ee
Equivalently, the imitation dynamics \eqref{eq:imitation-dynamics} may be rewritten as
\be\label{eq:imitation-dynamics-compact}\dot x=\diag(x)(F^T(x)-F(x))x\,,\ee 
where $F(x)=(f_{ij}(x))_{i,j}$ is a matrix-valued function on $\mc X$. 
\end{definition}\medskip 

Observe that, in order to satisfy \eqref{assumption:fij}, the functions $f_{ij}$s should clearly depend on the difference between the rewards $r_i(x)-r_j(x)$ in such a way that $f_{ij}(x)=f_{ji}(x)$, for every configuration $x$ such that $r_i(x)=r_j(x)$. In principle, these functions can possibly depend on the whole configuration $x$ in a non-trivial way. However, while our results hold true in such greater generality, we are mostly concerned  with the case where the functions $f_{ij}(x)$ only depend on the rewards' differences $r_i(x)-r_j(x)$, possibly in a different way for each different pair  of actions $i,j\in\mc A$. In fact, in this case, the considered imitation dynamics model makes minimal assumptions on the amount of information available to the players, i.e., they only know their own current action, the one of the other player met, and difference of their respective current rewards. In particular, players need not to know any other information about the game they are engaged in, such as, e.g., the current configuration of the population, the form of the reward functions,  or even the whole action space. 

{\remark This class of imitation dynamics generalize the ones considered in many papers \cite{Sandholm2010}, which satisfy
\be\label{assumption:sandholm}
r_i(x)\geq r_j(x)\iff f_{ki}(x)-f_{ik}(x)\geq f_{kj}(x)-f_{jk}(x)\,,
\ee
for every $i,j,k\in\mc A$. 
In fact, it is straightforward to check that \eqref{assumption:sandholm} is in general more restrictive than \eqref{assumption:fij}, that is obtained from \eqref{assumption:sandholm} in the case $k=j$. Notably, \eqref{assumption:sandholm} induces an ordering of the actions such that, when comparing two of them, the one with the larger reward should always result the more appealing to any third party, quite a restrictive condition that is not required in our more general formulation. Example \ref{ex:stochastic}, which follows, is a concrete example of a realistic situation in which our relaxed condition \eqref{assumption:fij} holds and \eqref{assumption:sandholm} does not.}\medskip

We now present two examples of imitation dynamics.

\begin{example}[Replicator Dynamics]\label{ex:replicator} In the case when $$f_{ij}(x)=\frac12\left(r_j(x)-r_i(x)\right),\qquad i,j\in\mc A\,,$$
or, equivalently, $F(x)=\frac12\left(\1r^T(x)-r(x)\1^T\right)$,
the imitation dynamics \eqref{eq:imitation-dynamics} reduces to the replicator equation 
\be\label{replicator} \dot x_i=x_i\left(r_i(x)-\ov r(x)\right)\,,\qquad i\in\mc A\,.\ee
Hence, imitation dynamics encompass and generalize the replicator equation, for which an extensive analysis has been developed, see, e.g., \cite{Weibull1995,Hofbauer1998,Taylor1978, Schuster1983}.\end{example}

\begin{example}[Stochastic Imitation Dynamics]\label{ex:stochastic}
Let 
\be\label{eq:stochastic} f_{ij}(x)=\frac{1}{2}+\frac{1}{\pi}\arctan(K_{ij}(r_j(x)-r_i(x)))\,,\ee
for $i,j\in\mc A$, where $K_{i,j}>0$. 
Such $[0,1]$-valued functions $f_{ij}(x)$ have an immediate interpretation as probabilities that players playing action $i$ switch to action $j$ when observing others playing such action $j$. Therefore these dynamics might be used when modeling mean-filed limits of stochastic imitation dynamics \cite{Benaim2003}. If the positive constants $K_{ij}$ are not all the same, the associated imitation dynamics may not satisfy \eqref{assumption:sandholm}, but still fit in our framework.
\end{example}\medskip

We now move on to discussing some general properties of imitation dynamics in continuous population games. To this aim, we first introduce some more notions related to Nash equilibria. For a nonempty subset of actions $\mc S\subseteq\mc A$, let $$\mc X_{\mc S}=\{x\in\mc X:\, x_i=0,\,\forall\,i\in\mc A\setminus\mc S\}$$ be the subset of configurations supported on $\mc S$ 
and let 
\be\label{eq:nash restricted}\mc N_{\mc S}=\left\{x\in\mc X_{\mc S}:\,x_i>0\Rightarrow r_i(x)\ge r_j(x),\,\forall\, j\in\mc S\right\}\ee
be the set of Nash equilibria of the population game restricted to $\mc S$.
Clearly, $\mc X_{\mc A}=\mc X$ and $\mc N_{\mc A}=\mc N$. 
Finally, we define the \emph{set of critical configurations} as
\be\label{def:Z}\mc Z=\bigcup_{\emptyset\ne\mc S\subseteq\mc A}\mc N_{\mc S}\,.\ee
Observe that $\mc Z$ includes the set of Nash equilibria $\mc N$ and can equivalently be characterized as
\be\label{Z-equivalent} \mc Z=\left\{x\in\mc X:\,\,x_i>0\Rightarrow r_i(x)=\ov r(x) \right\}\,.\ee
{\remark The set $\mc Z$ always includes the vertices $\delta^{(i)}$, $i\in\mc A$, of the  simplex $\mc X$. In fact, in the case when $|\mc A|=2$, the set of critical configurations consists just of the two vertices of $\mc X$ and the possible interior Nash equilibria of the game. For $|\mc A|\ge3$, the set of critical configurations $\mc Z$ includes, besides vertices of $\mc X$ and Nash equilibria of the game, all Nash equilibria of the sub-games obtained by restricting the action set to a non-trivial action subset $\mc S\subseteq\mc A$.}\medskip

Some basic properties of the imitation dynamics \eqref{eq:imitation-dynamics} are gathered in the following Lemma. These results are already proven in \cite{Sandholm2010} under the more restrictive condition on the dynamics. For the proof in our more general setting is included in the Appendix. 
\begin{lemma}\label{lemma:imitations-properties}
For any imitation dynamics \eqref{eq:imitation-dynamics} satisfying \eqref{assumption:fij}: 
\begin{enumerate}
\item[(i)] if $x(0)\in\mc X_{\mc S}$ for some nonempty subset of actions $\mc S\subseteq\mc A$, then $x(t)\in\mc X_{\mc S}$ for all $t\ge0$; 
\item[(ii)] if $x_i(0)>0$ for some $i\in\mc A$, then $x_i(t)>0$ for $t\ge0$; 
\item[(iii)] every restricted Nash equilibrium $x\in\mc Z$ is a rest point. 
\end{enumerate}
\end{lemma}
\section{Asymptotic Behavior of Imitation Dynamics  for Potential Population Games}\label{sec:asymptotics}
The main results of this work deal with the global asymptotic behavior of the imitation dynamics \eqref{eq:imitation-dynamics} for potential population games. Therefore, before presenting these results, we briefly introduce the notion of potential game \cite{Monderer1996} in the context of continuous population games.  
\begin{definition}
A population game with action set $\mc A$ and Lipschitz-continuous reward function vector $r:\mc X\to\R^{\mc A}$ is a \emph{potential population game} if there exists a \emph{potential function} $\Phi:\mc X\to\R$ that is continuous on $\mc X$, continuously differentiable in its interior, with gradient $\nabla\Phi(x)$ extendable by continuity to the boundary of $\mc X$, and such that 
\be\label{potential}
r_j(x)-r_i(x)=\frac{\partial}{\partial x_j}\Phi(x)-\frac{\partial}{\partial x_i}\Phi(x)\,,
\ee
for $i,j\in\mc A$, and almost every $x\in\mc X$.
\end{definition}\medskip

The asymptotic analysis of imitative dynamics for potential population games begins by proving that the potential function $\Phi(x)$ is never decreasing along trajectories of the imitation dynamics \eqref{eq:imitation-dynamics} and it is strictly increasing whenever $x$ does not belong to the set $\mc Z$ of critical configurations. This result, already known for more specific classes of dynamics \cite{Sandholm2010}, is thus generalized in the following result, whose proof is reported in the Appendix.


\begin{lemma}\label{lemma:Lyapunov} Let $r:\mc X\to\R^{\mc A}$ be the reward function vector of a potential population game with potential $\Phi:\mc X\to\R$. Then, every imitation dynamics \eqref{eq:imitation-dynamics} satisfying \eqref{assumption:fij} is such that 
\be\label{eq:Lyapunov}
\dot\Phi(x)=\nabla\Phi(x)\cdot\dot x\geq 0\,,\qquad\text{for all } x\in\mc X\,,
\ee
with equality if and only if $x\in\mc Z$, as defined in \eqref{def:Z}.  
\end{lemma}\medskip

An intuitive consequence of Lemma \ref{lemma:Lyapunov} and point (iii) of Lemma \ref{lemma:imitations-properties}  is that every imitation dynamics in a potential continuous population game has $\omega$-limit set coinciding with the set of critical configurations $\mc Z$. As we shall see, our main result, beyond formally proving this intuitive statement, consists in a significant refinement of it.  

Observe that, from \eqref{eq:nash} and \eqref{Z-equivalent}, the set $\mc B:=\mc Z\setminus\mc N$ of critical configurations that are not Nash equilibria satisfies
\be\label{def:B}\mc B=\left\{x\in\mc X:\,\,x_i>0\Rightarrow r_i(x)=\ov r(x)<r_*(x)\right\}\,.\ee
In other terms, critical configurations $x$ that are not Nash equilibria have the property that all actions played by a non-zero fraction of players in the population (i.e., those $i\in\mc A$ such that $x_i>0$) give the same average reward ($r_i(x)=\ov r(x)$), that is strictly less than the maximum reward ($\ov r(x)<r_*(x)$). This implies that $r_*(x)$ is necessarily achieved by some action that is not adopted by anyone, i.e., $\ov r(x)<r_*(x)=r_j(x)$ for some $j\in\mc A$ such that $x_j=0$.  

Notice that, in particular, $\mc B$ is a subset of the boundary of $\mc X$, since critical configurations that are not Nash equilibria necessarily have at least one zero entry.  The following result states that, in potential population games, every such configuration $\ov x\in\mc B$ has an interior neighborhood in $\mc X$ where the potential is strictly larger than in $\ov x$. This result is the main novelty of this work, being the key Lemma to prove global asymptotic stability results for the imitation dynamics.

\begin{lemma}\label{lemma:border-potential}
 Let $r:\mc X\to\R^{\mc A}$ be the reward function vector of a potential population game with potential $\Phi:\mc X\to\R$.
Let $\mc B=\mc Z\setminus\mc N$ be the set of critical configurations that are not Nash equilibria. Then, for every $\ov x\in\mc B$, there exists some $\varepsilon>0$ such that 
$\Phi(x)>\Phi(\bar x)$
for all $x\in\mc X$ such that 
\be\label{eq:xovx}||x-\bar x||<\eps\qquad\text{and}\qquad \sum_{i\in\mc A: r_i(\ov x)=r_*(\ov x)} x_i>0\,.\ee 
\end{lemma}\vspace{.2cm}
\begin{proof}
For  $\bar x\in\mc B$, let
$\mc I:=\{i\in\mc A:\,r_i(\ov x)=r_*(\ov x)\}$ and $\mc J=\mc A\setminus\mc I=\{i\in\mc A:\,r_i(\ov x)<r_*(\ov x)\}\,.$
From \eqref{potential}, 
$$
m:=\min_{i\in\mc I}\frac{\partial \Phi(\bar x)}{\partial x_i}-\max_{j\in\mc J}\frac{\partial \Phi(\bar x)}{\partial x_j}=r_*(\bar{x})-\max_{j\in\mc J}r_j(\bar x)>0\,.
$$
By continuity of $\nabla\Phi(x)$, there exists $\eps>0$ such that 
\be\label{eq:Phi-ij}
\min_{i\in\mc I}\frac{\partial \Phi(x)}{\partial x_i}-\max_{j\in\mc J}\frac{\partial \Phi(x)}{\partial x_j}\ge \frac m2\,,
\ee
for every $x\in\mc X$ such that $||x-\ov x||<\eps$. Then, fix any $x\in\mc X$ satisfying \eqref{eq:xovx}, let $z=x-\ov x$, and observe that 
\be\label{eq:sumz}a:=\sum_{i\in\mc I}z_i=-\sum_{j\in\mc J}z_j>0\,.\ee 
It then follows from \eqref{eq:Phi-ij} and \eqref{eq:sumz} that, for every point 
$$y(t)=\ov x+tz\,,\qquad t\in[0,1]\,,$$
along the segment joining $\ov x$ and $x$, one has that 
$$
\ba{rl}
\nabla\Phi(y(t))\cdot z=&\ds\sum_{i\in\mc I}z_i\frac{\partial}{\partial x_i}\Phi(y(t))-\sum_{j\in\mc J}z_j\frac{\partial}{\partial x_j}\Phi(y(t))\\
\ge&\ds a\min_{i\in\mc I}\frac{\partial}{\partial x_i}\Phi(y(t))-a\max_{j\in\mc J}\frac{\partial}{\partial x_j}\Phi(y(t))\\
\ge&\ds\frac{a m}2\,,
\ea
$$
so that 
$$\Phi(x)=\Phi(\ov x)+\int_0^1\left(\nabla\Phi(y(t))\cdot z\right)\,\de t\ge\Phi(\ov x)+\frac{am}2>\Phi(\ov x).$$

\end{proof}

In order to understand the novelty of this result we consider that in \cite{Sandholm2010}, where the stability of points in $\mc Z$ is analyzed for a subclass of imitation dynamics, it is proven that all the points in $\mc B$ are unstable, whereas a subset of the points in $\mc N$, coinciding with the local maximizers of $\Phi$ are stable. However, these two results deal with local stability and their mere combination is not sufficient to prove global asymptotic stability. On the contrary, our characterization of the instability of the rest points in $\mc B$ through the analysis of the value of the potential function in their neighborhood, paves the way for our main result, which characterizes the global asymptotic behavior of solutions of a broad class of imitation dynamics in potential population games. 

\begin{theorem}\label{main theorem}
Consider a potential population game with action set $\mc A$ and configuration space $\mc X$. 
Let $(x(t))_{t\ge0}$ be a solution of some imitation dynamics \eqref{eq:imitation-dynamics} satisfying \eqref{assumption:fij} and $$\mc S=\{i\in\mc A:\,x_i(0)>0\}$$ be the support of the initial configuration. Then, 
$$\lim_{t\to+\infty}\dist(x(t),\mc N_{\mc S})=0\,.$$ In particular, if $x_i(0)>0$ for every $i\in\mc A$, then $x(t)$ converges to the set $\mc N$ of Nash equilibria. 
\end{theorem}
\begin{proof} 
By Lemma \ref{lemma:imitations-properties} part (i) there is no loss of generality in assuming that $\mc S=\mc A$, i.e., $x_i(0)>0$ for every $i\in\mc A$. 
Let $r(x)$ be the reward vector function of the considered population game and let $\Phi(x)$ be a potential. Observe that $\Phi(x)$ is continuous over the compact configuration space $\mc X$, so that $$\Delta=\max\limits_{x\in\mc X}\Phi(x)-\min\limits_{x\in\mc X}\Phi(x)<+\infty\,.$$
Then, for every $t\ge0$ we have that 
$$\int_{0}^t\dot\Phi(x(s))\de s=\Phi(x(t))-\Phi(x(0))\le\Delta<+\infty$$
Since $\dot\Phi(x)\ge0$ for every $x\in\mc X$ by Lemma \ref{lemma:Lyapunov}, the above implies that $$\lim_{t\to+\infty}\dot\Phi(x(t))=0\,.$$
Then, continuity of $\dot\Phi(x)$ and the second part of Lemma \ref{lemma:Lyapunov} imply that $x(t)$ converges to the set $\mc Z$, as $t$ grows.

We are now left with proving that every solution $x(t)$ of an imitation dynamics with $x_i(0)>0$ for every $i\in\mc A$ approaches the subset $\mc N\subseteq\mc Z$ of Nash equilibria. 
By contradiction, let us assume that $\exists\,\eps>0$ such that $\forall\,t^*>0$ there exists some $t\ge t^*$ such that $\dist(x(t),\mc N)\ge\eps$. Since $x(t)$ approaches $\mc Z$ as $t$ grows large, this implies that for every $\eps>0$ and every large enough $t^*$   there exists  $t\ge t^*$ such that $\dist(x(t),\mc B)<\eps$. It follows that there exists a sequence of times $t_1\le t_2\le\ldots$ such that $\dist(x(t_n),\mc B)\stackrel{{n\to+\infty}}{\longrightarrow}0$. Since the configuration space $\mc X$ is compact, we may extract a converging subsequence $x(t_{n_k})$ with limit $\ov x\in\mc B$. Now, observe that Lemma \ref{lemma:imitations-properties} part (ii) implies that $x_i(t)>0$ for every action $i\in\mc A$. Then, Lemma \ref{lemma:border-potential} implies that there exists $k_0\ge1$ such that 
$$\Phi(x(t_k))>\Phi(\ov x)\,,\qquad \forall k\ge k_0\,.$$ Hence, the fact that $\Phi(x(t_k))$ is never decreasing as stated in Lemma \ref{lemma:Lyapunov}, would lead to 
$$\Phi(\ov x)=\lim_{k\to+\infty}\Phi(x(t_{n_k}))\ge\Phi(x(t_{n_{k_0}}))>\Phi(\ov x)\,,$$
a contradiction. Hence, $\lim\limits_{t\to+\infty}\dist(x(t),\mc N)=0$.  
\end{proof}

\section{Examples}\label{sec:examples}
In this section we present some applications of the results from Section \ref{sec:asymptotics}. For the imitation dynamics from Example \ref{eq:stochastic} (with all $K_{ij}$ sampled from independently and uniformly from $[0,1]$), we compare the analytical results obtained from Theorem \ref{main theorem} with some numerical simulations of the dynamics, in order to corroborate our theoretical results.

\subsection{Linear reward population games}\label{sec:binary}

We present some examples of binary games as in Example \ref{ex:linear-reward} and of pure coordination games.

\begin{example}[Binary linear reward games]\label{ex:binary}
First of all, it is straightforward to prove that all binary games are potential games. In fact, from a $2\times 2$ reward matrix $R$, as defined in \eqref{eq:R-binary},
we can immediately obtain a potential function, that is
\be\label{potential 2d}\phi(x)=\frac{1}{2}\left((a-c)x_1^2+(d-b)x_2^2\right).\ee

Notice that is not true that a generic linear reward game is potential, for example a ternary game such as Rock-Scissors-Paper is known to be not a potential game \cite{Sandholm2010}. 

In the following, three short examples of binary linear potential population games will be presented. Let us consider the following three reward matrices:
\be\label{matrices}
R^{(1)}=\left[\ba{ll}10 &0 \\ 8 &7\ea\right]\; R^{(2)}=\left[\ba{ll}0 &7 \\ 2 &6\ea\right]\; R^{(3)}=\left[\ba{ll}2&0\\3&1\ea\right],
\ee
Matrix $R^{(1)}$ leads to a coordination game. Trajectories converge to one of the three Nash equilibria: the global minimum of the potential function, attained in an interior point $\bar x$, and the two vertices of the simplex. Moreover, from Lemma \ref{lemma:Lyapunov}, we deduce that all trajectories with $x_1(0)<\bar x_1$ converge to $(0,1)$, all trajectories with $x_1(0)>\bar x_1$ converge to $(1,0)$, whereas $\bar x$ is an unstable equilibrium.

Matrix $R^{(2)}$ leads to an anti-coordination game, where the Nash equilibrium $\bar x$ is unique and it is an interior point. Therefore, if the support of the initial condition is $\mc A$, then Theorem \ref{main theorem} guarantees convergence to it.

Matrix $R^{(3)}$ leads to a game with a dominated action. In this case, the potential is a monotone increasing function in $x_2$. Therefore its maximum is attained in $\delta^{(2)}$, that is the only Nash equilibrium. Theorem \ref{main theorem} guarantees all trajectories with $\x_2(0)>0$ to converge to it.
Fig.~\ref{potentials} shows the plot of the potential functions of the three games and Fig.~\ref{2d} shows examples of trajectories of the imitation dynamics \eqref{eq:stochastic}.

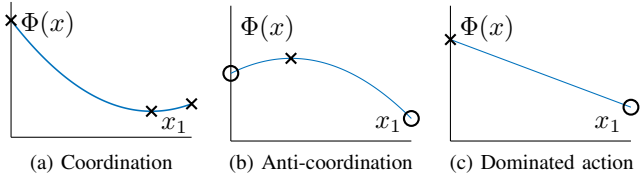
\begin{figure}
\centering
\subfloat[Coordination]{
%
%
\definecolor{mycolor1}{rgb}{0.00000,0.44700,0.74100}%
\definecolor{mycolor2}{rgb}{0.85000,0.32500,0.09800}%
\begin{tikzpicture}

\begin{axis}[%
 axis lines=middle,
 x   axis line style={-},
y   axis line style={-},
ytick style={draw=none},
ymajorticks=false,
xmajorticks=false,
width=\l cm,
height=\h cm,
at={(0.741in,0.457in)},
scale only axis,
xmin=0,
xmax=1,
xlabel={$x_1$},
ymin=0,
ymax=8.1,
ylabel={$\Phi(x)$},
axis background/.style={fill=white},
every axis x label/.style={
    at={(.78,.09)},
    anchor=west,
},
]

\addplot [color=mycolor1,solid,forget plot,semithick]
  table[row sep=crcr]{%
0	7\\
0.01	6.8609\\
0.02	6.7236\\
0.03	6.5881\\
0.04	6.4544\\
0.05	6.3225\\
0.06	6.1924\\
0.07	6.0641\\
0.08	5.9376\\
0.09	5.8129\\
0.1	5.69\\
0.11	5.5689\\
0.12	5.4496\\
0.13	5.3321\\
0.14	5.2164\\
0.15	5.1025\\
0.16	4.9904\\
0.17	4.8801\\
0.18	4.7716\\
0.19	4.6649\\
0.2	4.56\\
0.21	4.4569\\
0.22	4.3556\\
0.23	4.2561\\
0.24	4.1584\\
0.25	4.0625\\
0.26	3.9684\\
0.27	3.8761\\
0.28	3.7856\\
0.29	3.6969\\
0.3	3.61\\
0.31	3.5249\\
0.32	3.4416\\
0.33	3.3601\\
0.34	3.2804\\
0.35	3.2025\\
0.36	3.1264\\
0.37	3.0521\\
0.38	2.9796\\
0.39	2.9089\\
0.4	2.84\\
0.41	2.7729\\
0.42	2.7076\\
0.43	2.6441\\
0.44	2.5824\\
0.45	2.5225\\
0.46	2.4644\\
0.47	2.4081\\
0.48	2.3536\\
0.49	2.3009\\
0.5	2.25\\
0.51	2.2009\\
0.52	2.1536\\
0.53	2.1081\\
0.54	2.0644\\
0.55	2.0225\\
0.56	1.9824\\
0.57	1.9441\\
0.58	1.9076\\
0.59	1.8729\\
0.6	1.84\\
0.61	1.8089\\
0.62	1.7796\\
0.63	1.7521\\
0.64	1.7264\\
0.65	1.7025\\
0.66	1.6804\\
0.67	1.6601\\
0.68	1.6416\\
0.69	1.6249\\
0.7	1.61\\
0.71	1.5969\\
0.72	1.5856\\
0.73	1.5761\\
0.74	1.5684\\
0.75	1.5625\\
0.76	1.5584\\
0.77	1.5561\\
0.78	1.5556\\
0.79	1.5569\\
0.8	1.56\\
0.81	1.5649\\
0.82	1.5716\\
0.83	1.5801\\
0.84	1.5904\\
0.85	1.6025\\
0.86	1.6164\\
0.87	1.6321\\
0.88	1.6496\\
0.89	1.6689\\
0.9	1.69\\
0.91	1.7129\\
0.92	1.7376\\
0.93	1.7641\\
0.94	1.7924\\
0.95	1.8225\\
0.96	1.8544\\
0.97	1.8881\\
0.98	1.9236\\
0.99	1.9609\\
1	2\\
};

\addplot[only marks,mark=x,mark options={scale=1.4},text mark as node=true,thick] coordinates {(0.7778,1.5556)};
\addplot[only marks,mark=x,mark options={scale=1.4},text mark as node=true,black,thick] coordinates {(0,7)};
\addplot[only marks,mark=x,mark options={scale=1.4},text mark as node=true,black,thick] coordinates {(1,2)};

\end{axis}
\end{tikzpicture}%
}\quad\subfloat[Anti-coordination]{
%
%
\definecolor{mycolor1}{rgb}{0.00000,0.44700,0.74100}%
\definecolor{mycolor2}{rgb}{0.85000,0.32500,0.09800}%
\begin{tikzpicture}

\begin{axis}[%
 axis lines=middle,
 x   axis line style={-},
y   axis line style={-},
ytick style={draw=none},
ymajorticks=false,
xmajorticks=false,
width=\l cm,
height=\h  cm,
at={(0.741in,0.457in)},
scale only axis,
xmin=0,
xmax=1,
xlabel={$x_1$},
ymin=0,
ymax=3,
ylabel={$\Phi(x)$},
]

\addplot [color=mycolor1,solid,forget plot]
  table[row sep=crcr]{%
0	1.5\\
0.01	1.5197\\
0.02	1.5388\\
0.03	1.5573\\
0.04	1.5752\\
0.05	1.5925\\
0.06	1.6092\\
0.07	1.6253\\
0.08	1.6408\\
0.09	1.6557\\
0.1	1.67\\
0.11	1.6837\\
0.12	1.6968\\
0.13	1.7093\\
0.14	1.7212\\
0.15	1.7325\\
0.16	1.7432\\
0.17	1.7533\\
0.18	1.7628\\
0.19	1.7717\\
0.2	1.78\\
0.21	1.7877\\
0.22	1.7948\\
0.23	1.8013\\
0.24	1.8072\\
0.25	1.8125\\
0.26	1.8172\\
0.27	1.8213\\
0.28	1.8248\\
0.29	1.8277\\
0.3	1.83\\
0.31	1.8317\\
0.32	1.8328\\
0.33	1.8333\\
0.34	1.8332\\
0.35	1.8325\\
0.36	1.8312\\
0.37	1.8293\\
0.38	1.8268\\
0.39	1.8237\\
0.4	1.82\\
0.41	1.8157\\
0.42	1.8108\\
0.43	1.8053\\
0.44	1.7992\\
0.45	1.7925\\
0.46	1.7852\\
0.47	1.7773\\
0.48	1.7688\\
0.49	1.7597\\
0.5	1.75\\
0.51	1.7397\\
0.52	1.7288\\
0.53	1.7173\\
0.54	1.7052\\
0.55	1.6925\\
0.56	1.6792\\
0.57	1.6653\\
0.58	1.6508\\
0.59	1.6357\\
0.6	1.62\\
0.61	1.6037\\
0.62	1.5868\\
0.63	1.5693\\
0.64	1.5512\\
0.65	1.5325\\
0.66	1.5132\\
0.67	1.4933\\
0.68	1.4728\\
0.69	1.4517\\
0.7	1.43\\
0.71	1.4077\\
0.72	1.3848\\
0.73	1.3613\\
0.74	1.3372\\
0.75	1.3125\\
0.76	1.2872\\
0.77	1.2613\\
0.78	1.2348\\
0.79	1.2077\\
0.8	1.18\\
0.81	1.1517\\
0.82	1.1228\\
0.83	1.0933\\
0.84	1.0632\\
0.85	1.0325\\
0.86	1.0012\\
0.87	0.9693\\
0.88	0.9368\\
0.89	0.9037\\
0.9	0.87\\
0.91	0.8357\\
0.92	0.8008\\
0.93	0.7653\\
0.94	0.7292\\
0.95	0.6925\\
0.96	0.6552\\
0.97	0.6173\\
0.98	0.5788\\
0.99	0.5397\\
1	0.5\\
};

\addplot[only marks,mark=x,mark options={scale=1.4},text mark as node=true,black,thick] coordinates {(0.3333,1.8333)};
\addplot[only marks,mark=o,mark options={scale=1.4},text mark as node=true,thick] coordinates {(0,1.5)};
\addplot[only marks,mark=o,mark options={scale=1.4},text mark as node=true,thick] coordinates {(1,0.5)};

\end{axis}
\end{tikzpicture}%
}\quad\subfloat[Dominated action]{
%
%
\definecolor{mycolor1}{rgb}{0.00000,0.44700,0.74100}%
\definecolor{mycolor2}{rgb}{0.85000,0.32500,0.09800}%
\begin{tikzpicture}

\begin{axis}[%
 axis lines=middle,
 x   axis line style={-},
y   axis line style={-},
ytick style={draw=none},
ymajorticks=false,
xmajorticks=false,
width=\l cm,
height=\h cm,
at={(0.741in,0.457in)},
scale only axis,
xmin=0,
xmax=1,
xlabel={$x_1$},
ymin=0,
ymax=4,
ylabel={$\Phi(x)$},
axis background/.style={fill=white},
]

\addplot [color=mycolor1,solid,forget plot]
  table[row sep=crcr]{%
0	3\\
0.01	2.98\\
0.02	2.96\\
0.03	2.94\\
0.04	2.92\\
0.05	2.9\\
0.06	2.88\\
0.07	2.86\\
0.08	2.84\\
0.09	2.82\\
0.1	2.8\\
0.11	2.78\\
0.12	2.76\\
0.13	2.74\\
0.14	2.72\\
0.15	2.7\\
0.16	2.68\\
0.17	2.66\\
0.18	2.64\\
0.19	2.62\\
0.2	2.6\\
0.21	2.58\\
0.22	2.56\\
0.23	2.54\\
0.24	2.52\\
0.25	2.5\\
0.26	2.48\\
0.27	2.46\\
0.28	2.44\\
0.29	2.42\\
0.3	2.4\\
0.31	2.38\\
0.32	2.36\\
0.33	2.34\\
0.34	2.32\\
0.35	2.3\\
0.36	2.28\\
0.37	2.26\\
0.38	2.24\\
0.39	2.22\\
0.4	2.2\\
0.41	2.18\\
0.42	2.16\\
0.43	2.14\\
0.44	2.12\\
0.45	2.1\\
0.46	2.08\\
0.47	2.06\\
0.48	2.04\\
0.49	2.02\\
0.5	2\\
0.51	1.98\\
0.52	1.96\\
0.53	1.94\\
0.54	1.92\\
0.55	1.9\\
0.56	1.88\\
0.57	1.86\\
0.58	1.84\\
0.59	1.82\\
0.6	1.8\\
0.61	1.78\\
0.62	1.76\\
0.63	1.74\\
0.64	1.72\\
0.65	1.7\\
0.66	1.68\\
0.67	1.66\\
0.68	1.64\\
0.69	1.62\\
0.7	1.6\\
0.71	1.58\\
0.72	1.56\\
0.73	1.54\\
0.74	1.52\\
0.75	1.5\\
0.76	1.48\\
0.77	1.46\\
0.78	1.44\\
0.79	1.42\\
0.8	1.4\\
0.81	1.38\\
0.82	1.36\\
0.83	1.34\\
0.84	1.32\\
0.85	1.3\\
0.86	1.28\\
0.87	1.26\\
0.88	1.24\\
0.89	1.22\\
0.9	1.2\\
0.91	1.18\\
0.92	1.16\\
0.93	1.14\\
0.94	1.12\\
0.95	1.1\\
0.96	1.08\\
0.97	1.06\\
0.98	1.04\\
0.99	1.02\\
1	1\\
};

\addplot[only marks,mark=x,mark options={scale=1.4},text mark as node=true,black,thick] coordinates {(0,3)};
\addplot[only marks,mark=o,mark options={scale=1.4},text mark as node=true,thick] coordinates {(1,1)};

\end{axis}
\end{tikzpicture}%
}
\caption{Potentials of the games from Example \ref{ex:binary}. Crosses are Nash equilibria, circles are Nash equilibria for restricted games.}
\label{potentials}
\end{figure}

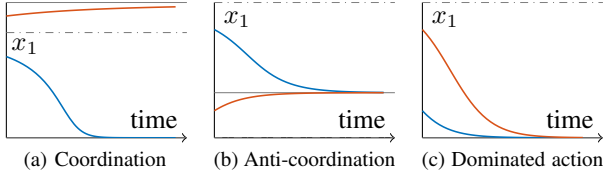
\begin{figure}
\centering
\subfloat[Coordination]{
%
%
\definecolor{mycolor1}{rgb}{0.00000,0.44700,0.74100}%
\definecolor{mycolor2}{rgb}{0.85000,0.32500,0.09800}%
\begin{tikzpicture}

\begin{axis}[%
 axis lines=middle,
 x   axis line style={->},
y   axis line style={-},
xtick style={draw=none},
xmajorticks=false,
ymajorticks=false,
width=\l  cm,
height=\h cm,
at={(0.741in,0.457in)},
scale only axis,
xmin=0,
xmax=16,
xlabel={time},
ymin=0,
ymax=1,
ylabel={$x_1$},
y label style={at={(.23,.55)},anchor=south east},
axis background/.style={fill=white},
]

\addplot [color=gray,forget plot,dashdotted]
  table[row sep=crcr]{%
0	0.7778\\
16	0.7778\\
};

\addplot [color=gray,forget plot]
  table[row sep=crcr]{%
0	1\\
16	1\\
};
\addplot [color=gray,forget plot]
  table[row sep=crcr]{%
0	0\\
16	0\\
};

\addplot [color=mycolor1,solid,forget plot, semithick]
  table[row sep=crcr]{%
0	0.6\\
0.375	0.587854152315965\\
0.75	0.574327160583299\\
1.125	0.559189297276333\\
1.5	0.542163428862961\\
1.875	0.522921531144345\\
2.25	0.501071314320051\\
2.625	0.476163816258384\\
3	0.447696907237842\\
3.375	0.415162397674111\\
3.75	0.378043347260584\\
4.125	0.336185568807423\\
4.5	0.289969467534735\\
4.875	0.240435953126461\\
5.25	0.190314312410416\\
5.625	0.142895936960525\\
6	0.10144218679145\\
6.17929857183361	0.0845659440391466\\
6.35859714366722	0.0697389392191834\\
6.53789571550082	0.0569587287823689\\
6.71719428733443	0.0461287071190648\\
6.89649285916804	0.037072933507803\\
7.07579143100165	0.0296046664829181\\
7.25509000283525	0.0235264365239166\\
7.43438857466886	0.0186220524446732\\
7.60380125981313	0.0148732392713943\\
7.77321394495739	0.0118480747988482\\
7.94262663010166	0.00942551924480847\\
8.11203931524592	0.00748829397932116\\
8.2619385720862	0.00609735146387477\\
8.41183782892647	0.00496058166991582\\
8.56173708576675	0.004035004252591\\
8.71163634260702	0.00328075223321611\\
8.85598890644812	0.00268506705446508\\
9.00034147028921	0.00219682144454616\\
9.14469403413031	0.00179769435564019\\
9.28904659797141	0.00147090627222816\\
9.43119466105946	0.00120624019980385\\
9.57334272414751	0.000989077097499046\\
9.71549078723556	0.00081128501005097\\
9.85763885032361	0.000665453604704595\\
9.99883595063323	0.000546182634858482\\
10.1400330509429	0.000448274692362665\\
10.2812301512525	0.000368066208289462\\
10.4224272515621	0.000302225941277697\\
10.5751482376012	0.000244005313455383\\
10.7278692236403	0.000197003308421787\\
10.8805902096795	0.000159170628673856\\
11.0333111957186	0.000128624199626843\\
11.2110487855808	0.000100210145591512\\
11.388786375443	7.80771110917306e-05\\
11.5665239653051	6.09506692278538e-05\\
11.7442615551673	4.76083932948579e-05\\
11.9528940508747	3.54697585733465e-05\\
12.1615265465821	2.64263200195596e-05\\
12.3701590422895	1.97966538550573e-05\\
12.5787915379969	1.48600238336305e-05\\
12.8302144024137	1.03776994924765e-05\\
13.0816372668305	7.2422717643738e-06\\
13.3330601312473	5.15084570737207e-06\\
13.5844829956641	3.69499297035053e-06\\
13.8967608115137	2.32238427132525e-06\\
14.2090386273634	1.44859886880036e-06\\
14.521316443213	9.85743011313304e-07\\
14.8335942590626	7.0221171363898e-07\\
14.8751956942969	6.6248157763588e-07\\
14.9167971295313	6.24999399039145e-07\\
14.9583985647656	5.89638172592153e-07\\
15	5.56277601123046e-07\\
};
\addplot [color=mycolor2,solid,forget plot, semithick]
  table[row sep=crcr]{%
0	0.9\\
0.375	0.903638439626736\\
0.75	0.90712837564807\\
1.125	0.910470228493903\\
1.5	0.91366574488509\\
1.875	0.916717828361863\\
2.25	0.919630186348126\\
2.625	0.922407125741202\\
3	0.925053432948752\\
3.375	0.927574269028678\\
3.75	0.929974947972696\\
4.125	0.932260862120661\\
4.5	0.934437426344045\\
4.875	0.936510026963475\\
5.25	0.938483912865734\\
5.625	0.940364188210956\\
6	0.942155795213555\\
6.375	0.943863494164823\\
6.75	0.945491820278015\\
7.125	0.947045101977466\\
7.5	0.948527461402542\\
7.875	0.949942809358124\\
8.25	0.951294833724744\\
8.625	0.952587023257969\\
9	0.953822674335305\\
9.375	0.95500489204754\\
9.75	0.95613659161301\\
10.125	0.957220520009432\\
10.5	0.958259263804434\\
10.875	0.959255252244968\\
11.25	0.960210762962174\\
11.625	0.961127939385092\\
12	0.962008797733028\\
12.375	0.962855230357326\\
12.75	0.963669012081406\\
13.125	0.964451813495156\\
13.5	0.965205206614833\\
13.875	0.965930667844275\\
14.25	0.966629583625204\\
14.625	0.967303260346848\\
15	0.967952928726937\\
};

\end{axis}
\end{tikzpicture}%
}\quad\subfloat[Anti-coordination]{
%
%
\definecolor{mycolor1}{rgb}{0.00000,0.44700,0.74100}%
\definecolor{mycolor2}{rgb}{0.85000,0.32500,0.09800}%
\begin{tikzpicture}

\begin{axis}[%
 axis lines=middle,
 x   axis line style={->},
y   axis line style={-},
xtick style={draw=none},
xmajorticks=false,
ymajorticks=false,
width=\l  cm,
height=\h cm,
at={(0.741in,0.457in)},
scale only axis,
xmin=0,
xmax=16,
xlabel={time},
ymin=0,
ymax=1,
ylabel={$x_1$},
axis background/.style={fill=white},
]

\addplot [color=gray,forget plot]
  table[row sep=crcr]{%
0	0.3333333\\
16	0.3333333\\
};
\addplot [color=gray,dashdotted,forget plot]
  table[row sep=crcr]{%
0	1\\
16	1\\
};
\addplot [color=gray,dashdotted,forget plot]
  table[row sep=crcr]{%
0	0\\
16	0\\
};

\addplot [color=mycolor1,solid,forget plot, semithick]
  table[row sep=crcr]{%
0	0.8\\
0.375	0.782246875254846\\
0.75	0.762500537434247\\
1.125	0.740700864865957\\
1.5	0.716881500818344\\
1.875	0.691186730592521\\
2.25	0.663936819777443\\
2.625	0.635595534647423\\
3	0.606754413301736\\
3.375	0.57809574939656\\
3.75	0.550300183078833\\
4.125	0.523951304720861\\
4.5	0.499497642919182\\
4.875	0.477236795833303\\
5.25	0.457285845241629\\
5.625	0.439646041069562\\
6	0.424218356218975\\
6.375	0.41082609724509\\
6.75	0.399278222398504\\
7.125	0.389374896978457\\
7.5	0.380910849780114\\
7.875	0.373687639183361\\
8.25	0.367541562453402\\
8.625	0.362323436456965\\
9	0.357894820889633\\
9.375	0.354134162851307\\
9.75	0.350946793126413\\
10.125	0.34824851789897\\
10.5	0.345962961806353\\
10.875	0.344024746249163\\
11.25	0.342383771067007\\
11.625	0.340995706640914\\
12	0.33982057567876\\
12.375	0.338824390997534\\
12.75	0.337981223715517\\
13.125	0.337268159721744\\
13.5	0.336664566884047\\
13.875	0.336152936829198\\
14.25	0.335719928359632\\
14.625	0.335353755225233\\
15	0.335043809355178\\
};
\addplot [color=mycolor2,solid,forget plot, semithick]
  table[row sep=crcr]{%
0	0.2\\
0.196241127461686	0.209789842380964\\
0.392482254923373	0.219049350670214\\
0.588723382385059	0.227764769379239\\
0.784964509846745	0.235933857171313\\
1.15996450984675	0.250058735658594\\
1.53496450984675	0.262337765120444\\
1.90996450984675	0.272926006897397\\
2.28496450984675	0.282006463552122\\
2.65996450984675	0.289773655120083\\
3.03496450984675	0.296393508622826\\
3.40996450984675	0.302020250580329\\
3.78496450984675	0.306798984261352\\
4.15996450984674	0.310858838349881\\
4.53496450984674	0.314301004581007\\
4.90996450984674	0.317215699251859\\
5.28496450984674	0.319684953179719\\
5.65996450984674	0.3217791655033\\
6.03496450984674	0.323552357184634\\
6.40996450984674	0.325052347217004\\
6.78496450984674	0.326322280661567\\
7.15996450984674	0.32739885806691\\
7.53496450984674	0.328310087323475\\
7.90996450984674	0.32908072205099\\
8.28496450984674	0.329733053195693\\
8.65996450984674	0.330285999250083\\
9.03496450984674	0.330753977047663\\
9.40996450984674	0.331149723379109\\
9.78496450984674	0.33148470215249\\
10.1599645098467	0.331768637148413\\
10.5349645098467	0.332008935578941\\
10.9099645098467	0.332212140814674\\
11.2849645098467	0.33238414152927\\
11.6599645098467	0.332529931774454\\
12.0349645098467	0.332653315455033\\
12.4099645098467	0.33275765277028\\
12.7849645098467	0.332845967611068\\
13.1599645098467	0.332920824367461\\
13.5349645098467	0.332984176251388\\
13.9099645098467	0.33303774863101\\
14.2849645098467	0.333083094175546\\
14.4637233823851	0.333102206385647\\
14.6424822549234	0.333119858748892\\
14.8212411274617	0.333136162417919\\
15	0.333151220926029\\
};

\end{axis}
\end{tikzpicture}%
}\quad\subfloat[Dominated action]{
%
%
\definecolor{mycolor1}{rgb}{0.00000,0.44700,0.74100}%
\definecolor{mycolor2}{rgb}{0.85000,0.32500,0.09800}%
\begin{tikzpicture}

\begin{axis}[%
 axis lines=middle,
 x   axis line style={->},
y   axis line style={-},
xtick style={draw=none},
xmajorticks=false,
ymajorticks=false,
width=\l cm,
height=\h cm,
at={(0.741in,0.457in)},
scale only axis,
xmin=0,
xmax=9,
xlabel={time},
ymin=0,
ymax=1,
ylabel={$x_1$},
axis background/.style={fill=white},
]

\addplot [color=gray,forget plot, dashdotted]
  table[row sep=crcr]{%
0	1\\
9	1\\
};

\addplot [color=gray,forget plot]
  table[row sep=crcr]{%
0	0\\
9	0\\
};

\addplot [color=mycolor1,solid,forget plot, semithick]
  table[row sep=crcr]{%
0	0.2\\
0.0627971607877395	0.190141420966944\\
0.125594321575479	0.180659043504348\\
0.188391482363218	0.171549396094173\\
0.251188643150958	0.162807842158672\\
0.451188643150958	0.137344698999804\\
0.651188643150958	0.115310914675103\\
0.851188643150958	0.0964243011882164\\
1.05118864315096	0.0803562885294334\\
1.25118864315096	0.0667555066951454\\
1.45118864315096	0.055316861235003\\
1.65118864315096	0.0457515999582237\\
1.85118864315096	0.0377778920061049\\
2.05118864315096	0.0311378220108504\\
2.25118864315096	0.0256337844503955\\
2.45118864315096	0.0210885316361108\\
2.65118864315096	0.0173367190825872\\
2.85118864315096	0.0142358208251876\\
3.05118864315096	0.0116832121064848\\
3.25118864315096	0.00958787798589877\\
3.45118864315096	0.0078662039649815\\
3.65118864315096	0.00644799013954915\\
3.85118864315096	0.00528427957487145\\
4.05118864315096	0.00433169223078775\\
4.25118864315096	0.00355060763048142\\
4.4495841042762	0.00291284186428498\\
4.64797956540145	0.00238944701850369\\
4.8463750265267	0.00196082496370183\\
5.04477048765194	0.00160913503893626\\
5.24215854945918	0.00132096967319656\\
5.43954661126641	0.00108439561195135\\
5.63693467307364	0.000890562032717785\\
5.83432273488088	0.000731424730578565\\
6.03127735038799	0.000600606770371298\\
6.22823196589509	0.00049319423559615\\
6.4251865814022	0.00040516936953149\\
6.62214119690931	0.000332883152330183\\
6.82214119690931	0.000272489446995377\\
7.02214119690931	0.000223059839948552\\
7.22214119690931	0.000182687497348196\\
7.42214119690931	0.000149637995225258\\
7.56660589768198	0.00012950682998358\\
7.71107059845466	0.000112084898777454\\
7.85553529922733	9.70135147166901e-05\\
8	8.39693063616152e-05\\
};
\addplot [color=mycolor2,solid,forget plot, semithick]
  table[row sep=crcr]{%
0	0.8\\
0.2	0.766079113622929\\
0.4	0.728354581186598\\
0.6	0.687030909906665\\
0.8	0.642509029896568\\
1	0.595378388321197\\
1.2	0.546431569730729\\
1.4	0.496571846675447\\
1.6	0.446769583953069\\
1.8	0.39801578829588\\
2	0.351209365443602\\
2.2	0.307098059597752\\
2.4	0.266249897235944\\
2.6	0.229041142523597\\
2.8	0.19564081756887\\
3	0.166064718388434\\
3.2	0.140186535556139\\
3.4	0.117762049582873\\
3.6	0.0985098107047544\\
3.8	0.0821204410539077\\
4	0.0682570374991805\\
4.2	0.0565772808512247\\
4.4	0.0467947986409069\\
4.6	0.0386436306943371\\
4.8	0.0318680392701534\\
5	0.0262381281983038\\
5.2	0.0215807796692446\\
5.4	0.0177413957474716\\
5.6	0.0145764524184118\\
5.8	0.0119631874063856\\
6	0.00981400173226846\\
6.2	0.00805125194738344\\
6.4	0.00660373484354001\\
6.6	0.00541187798259405\\
6.8	0.00443432190435516\\
7	0.00363441359981641\\
7.2	0.00297870079582192\\
7.4	0.00243947232101753\\
7.6	0.00199774189559099\\
7.8	0.00163667099754229\\
8	0.00134092225414413\\
};

\end{axis}
\end{tikzpicture}%
}
\caption{Trajectories of the imitation dynamics \eqref{eq:stochastic} for the games from Example \ref{ex:binary}. Solid lines are asymptotically stable equilibria, dotted lines are unstable.}
\label{2d}
\end{figure}
\end{example}

\begin{example}[Pure coordination games]\label{ex:coordination}
Another class of potential games are the linear reward pure coordination games \cite{Cooper1999}, in which the reward matrix $R$ is a diagonal positive (entry-wise) matrix. It is straightforward to check that a potential function is given by
\be\Phi(x)=\frac{1}{2}\sum_{i=1}^m R_{ii}x_i^2.\ee

Being $\Phi(x)$ convex, its minimum is attained in an interior point $\bar x$, and all the vertices of the simplex are local maxima of the potential function. All the other critical points are minima of the potential subject to belong to the boundaries. All of these points are Nash equilibria. Therefore Theorem \ref{main theorem} guarantees asymptotic convergence to them. In the ternary case $\abs{\mc A}=3$, a complete analysis can be carried out. Without any loss in generality, we can set $R_{11}=1$ and name $R_{22}=b$ and $R_{33}=c$. Then, analyzing $\Phi(x)=\frac{1}{2}\left(x_1^2+bx_2^2+cx_3^2\right)$, we explicitly compute the seven Nash equilibria: $\delta^{(1)}$, $\delta^{(2)}$, and $\delta^{(3)}$, the global minimum of the potential
$$
\bar x=\left(\frac{bc}{b+c+bc},\frac{c}{b+c+bc},\frac{b}{b+c+bc}\right),$$
and the three minima on the boundary of $\mc X$:
$$\bar x^{(1)}=\displaystyle\frac{(0,c,b)}{b+c},\quad\bar x^{(2)}=\frac{(c,0,1)}{c+1},\quad\bar x^{(3)}=\frac{(b,1,0)}{b+1}.
$$
Through Lemma \ref{lemma:Lyapunov} we conclude that $\bar x$ is an unstable node, the three points on the boundaries $\bar x^{(1)}$, $\bar x^{(2)}$, and $\bar x^{(3)}$ are saddle points, whose stable manifolds actually divide the basins of attraction of the three asymptotically stable nodes $\delta^{(1)}$, $\delta^{(2)}$, $\delta^{(3)}$. Fig. \ref{fig:coordination} shows two examples of potential and velocity plots of the imitation dynamics \eqref{eq:stochastic} for these games.
\begin{figure}
\centering
\subfloat[$b=2$, $c=3$]{\includegraphics[scale=.16]{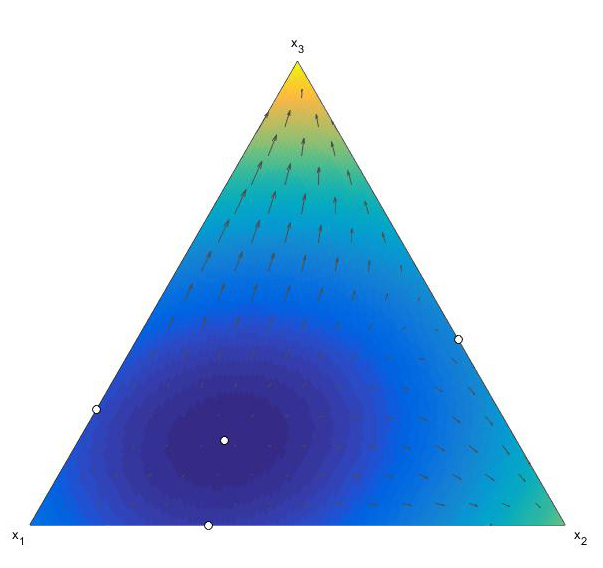}}\quad\subfloat[$b=0.2$, $c=5$]{\includegraphics[scale=.16]{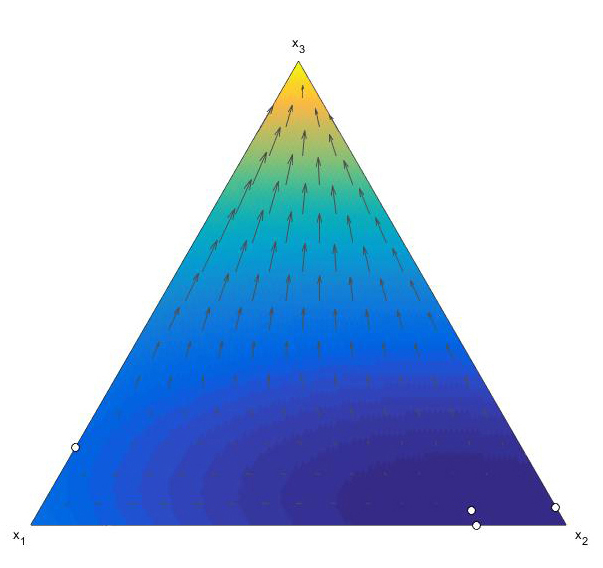}}
\caption{Potential of the pure coordination games from Example \ref{ex:coordination} and velocity plot of imitation dynamics \eqref{eq:stochastic} from Example \ref{ex:stochastic} for them. The unstable nodes and the saddle points are denoted by white circles.}
\label{fig:coordination}
\end{figure}
\end{example}

\subsection{Congestion games}

Another important class of potential games are congestion games \cite{Rosenthal1973, Monderer1996}. Let $\mc A=\{1,\dots, l\}$ be a set of resources and $A\in\{0,1\}^{l\times m}$ be the adjacency matrix of a bipartite graph connecting agents with resources. Let us introduce $l$ continuous functions, collected in a vector $\psi(\cdot)=(\psi_1(\cdot),\dots,\psi_l(\cdot))$, where the generic $\psi_k(y)$ is the reward for agents that use resource $k$, when the resource is used by a fraction $y$ of agents. The reward vector functions for these games are simply $r(x)=A^T\psi(Ax)$ and a straightforward computation shows that congestion games are always potential games, with \be\Phi(x)=\sum_{k=1}^{k=l}\Psi_k((Ax)_k),\ee
where $\Psi_k$ is an anti-derivative of $\psi_k$. 

Often, the functions $\psi_k$s represent a cost for the use of the resources, so they are monotone decreasing functions. In this case, the potential function $\Phi(x)$ is concave, possessing a global maximum $\bar x$, that is the only Nash equilibrium of the game. Depending on $A$, $\bar x$ can be an interior point, or it can belong to the boundary of the simplex. As the other critical points are considered, $\delta^{(i)}$ are minima of the potential, whereas local maxima are present on the boundary, that are Nash equilibria for restricted games. Theorem \ref{main theorem} guarantees therefore that trajectories with $x(0)>0$ (entry-wise) converge to $\bar x$, that is an asymptotically stable node. The Nash equilibria of the restricted games are saddle points (i.e., stable on the respective boundaries)  and the vertices that are not in one of the previous set are unstable nodes. Fig. \ref{fig:congestion} shows the velocity plots of the imitation dynamics for the following examples of congestion games.

\begin{example}[Exponential costs game]\label{ex:exponential}
Let $A=I$ and let the cost be $\psi_i(x_i)=\exp(c_1x_1)$, for $c_i>0$. Then, the maximum of the potential 
$\Phi(x)=-\sum_{i=1}^{m}\frac{1}{c_i}\exp(-c_ix_i)$
is achieved in an interior point $\bar x$, that is the unique Nash equilibrium of the game.
\end{example}

\begin{figure}
\centering
\subfloat[Ex. \ref{ex:exponential}]{\includegraphics[scale=.18]{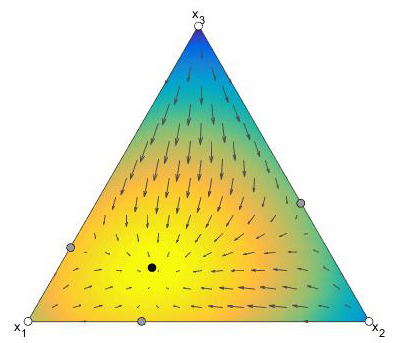}}\quad\subfloat[\eqref{eq:phi1} from Ex. \ref{ex:dominated}]{\includegraphics[scale=.18]{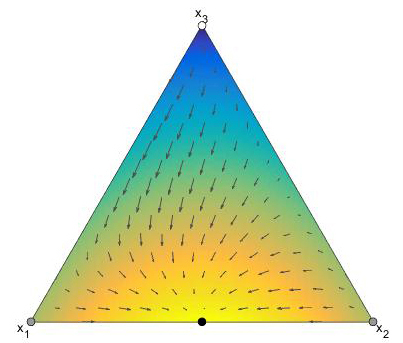}}\quad\subfloat[\eqref{eq:phi2} from Ex. \ref{ex:dominated}]{\includegraphics[scale=.18]{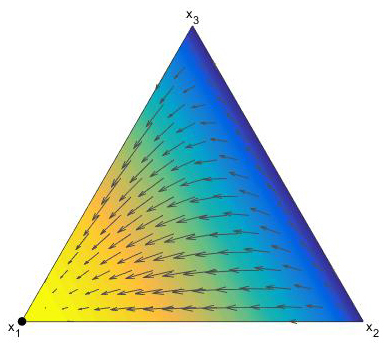}}
\caption{Potential of the congestion games from Example \ref{ex:exponential} (with $m=3$, $c_1=1$, $c_2=2$, and $c_3=3$) and from Example \ref{ex:dominated}, respectively, and velocity plot of imitation dynamics \eqref{eq:stochastic} from Example \eqref{ex:stochastic} for them. The unstable nodes are denoted by white circles, saddle points by gray circles, and black circles denote the only asymptotically stable equilibrium.}
\label{fig:congestion}
\end{figure}

\begin{example}[Dominated strategy]\label{ex:dominated}
We construct now two examples of congestion games in which the Nash equilibrium of the dynamics is on the boundary. Let $l=2$, $\psi_i(y)=-y$, and let us consider the following two adjacency matrices:
\be A_1=\left( \begin{array}{ccc}
1 & 0 & 1 \\
0 & 1 & 1 \end{array} \right)\quad\quad A_2=\left( \begin{array}{ccc}
1 & 1 & 1 \\
0 & 1 & 1 \end{array} \right).
\ee
The potential functions are, respectively:
\be\label{eq:phi1}\Phi_1(x)=-\frac{1}{2}\left((x_1+x_3)^2+(x_2+x_3)^2 \right)\ee\be\label{eq:phi2}\Phi_2(x)=-\frac{1}{2}(x_2+x_3)^2.\ee
As $A_1$ is considered, the Nash equilibrium of the dynamics is in $\bar x=(1/2,1/2,0)$, whereas $A_2$ has its Nash equilibrium in the vertex $\delta^{(1)}$. 
\end{example}

\section{Conclusion and Further Work}\label{sec:conclusions}

In this work we analyzed the asymptotic behavior of imitation dynamics in potential population games, proving convergence of the dynamics to the set of Nash equilibria of the sub-game restricted to the set of actions used in the initial configuration of the population. This results strengthen the state of the art, both ensuring global stability to the Nash equilibria, and generalizing the result to a class of dynamics that encompasses the replicator dynamics and the class of imitation dynamics considered in many previous works. 

The main research lines arising from this work point in two directions. On the one hand, taking advantage on the techniques developed in this work, our analysis has to be extended to the case in which the population is not fully mixed and agents interact on a non-complete communication network, similar to what have been done for other learning mechanisms, such as the replicator and logit choice \cite{Marden2012,Barreiro-Gomez2016}, or to cases in which the learning process interacts with the dynamics of a physical system \cite{Como2013}. On the other hand, stochasticity in the revising of the agent's opinion should be included into the imitation dynamics. This leads to model imitation dynamics with Markovian stochastic processes, paving the way for the study of several interesting open problems in the relationships between the asymptotic behavior of the new stochastic process and the one of the deterministic process analyzed in this work.

\section*{Appendix}
\setcounter{theorem}{1}
\begin{lemma}\label{lemma:imitations-properties=app}
For any imitation dynamics \eqref{eq:imitation-dynamics} satisfying \eqref{assumption:fij}: 
\begin{enumerate}
\item[(i)] if $x(0)\in\mc X_{\mc S}$ for some nonempty subset of actions $\mc S\subseteq\mc A$, then $x(t)\in\mc X_{\mc S}$ for all $t\ge0$; 
\item[(ii)] if $x_i(0)>0$ for some action $i\in\mc A$, then $x_i(t)>0$ for every $t\ge0$; 
\item[(iii)] every restricted Nash equilibrium $x\in\mc Z$ is a rest point. 
\end{enumerate}
\end{lemma}
\begin{proof}
\begin{enumerate}
\item[(i)] It follows from the fact that any solution of \eqref{eq:imitation-dynamics} with $x_i(0)=0$ for some $i\in\mc A$ is such that $x_i(t)=0$, $\forall\,t\ge0$. 
\item[(ii)] Since $\dot x_i(t)\ge-C_ix_i(t)$, where $C_i=|\mc A|\max\{|r_i(x):\,x\in\mc X|\}$, Gronwall's inequality implies that $x_i(t)\ge x_i(0)e^{-C_it}>0$. 
\item[(iii)] For every $x\in\mc Z$ and $i,j\in\mc A$, one has that $x_ix_j(r_i(x)-r_j(x))=0$. Then, \eqref{assumption:fij} implies that $x_ix_j(f_{ij}(x)-f_{ji}(x)))=0$. 
\end{enumerate}
\end{proof}\medskip
\setcounter{theorem}{3}
\begin{lemma}\label{lemma:Lyapunov-app} Let $r:\mc X\to\R^{\mc A}$ be the reward function vector of a potential population game with potential function $\Phi:\mc X\to\R$. Then, every imitation dynamics \eqref{eq:imitation-dynamics} satisfying \eqref{assumption:fij} is such that 
\be\label{eq:Lyapunov}
\dot\Phi(x)=\nabla\Phi(x)\cdot\dot x\geq 0\,,\qquad\text{for all } x\in\mc X\,,
\ee
with equality if and only if $x\in\mc Z$, as defined in \eqref{def:Z}.  
\end{lemma}
\begin{proof} For every $x\in\mc X$, we have 
\be\label{eq:chain}
\begin{array}{l}
\dot\Phi(x)=\nabla\Phi(x)\cdot\dot x\\[7pt]
=\ds\nabla\Phi(x)\cdot\diag(x)(F'(x)-F(x))x\\[7pt]
=\ds\sum_{i,j\in\mc A}\frac{\partial\Phi(x)}{\partial x_i}x_ix_j\left(f_{ji}(x)-f_{ij}(x)\right)\\[7pt]
=\ds\frac12\sum_{i,j\in\mc A}x_ix_j\left(\frac{\partial\Phi(x)}{\partial x_i}-\frac{\partial\Phi(x)}{\partial x_j}\right)\left(f_{ji}(x)-f_{ij}(x)\right)\\[7pt]
=\ds\frac12\sum_{i,j\in\mc A}x_ix_j\left(r_i(x)-r_j(x)\right)\left(f_{ji}(x)-f_{ij}(x)\right)\,,
\end{array}
\ee
where the last identity follows from \eqref{potential}. It now follows from property \eqref{assumption:fij} of the imitation dynamics that, $\forall\,i,j\in\mc A$, 
$$\left(r_i(x)-r_j(x)\right)\left(f_{ji}(x)-f_{ij}(x)\right)\ge0\,.$$
Being all entries of a configuration $x\in\mc X$ are non-negative, 
$$x_ix_j\left(r_i(x)-r_j(x)\right)\left(f_{ji}(x)-f_{ij}(x)\right)\ge0.$$
Combining the above with \eqref{eq:chain}, we get that $\dot\Phi(x)\ge0$ (thus proving \eqref{eq:Lyapunov}).
%
Finally, $\dot\Phi(x)=0$ if and only if all the terms 
\be\label{eq:Z-equiv}x_ix_j(r_i(x)-r_j(x))\left(f_{ji}(x)-f_{ij}(x)\right)=0\,,\ee
$\forall\,i,j\in\mc A$. Using again \eqref{assumption:fij}, we have that 
$$(r_i(x)-r_j(x))\left(f_{ji}(x)-f_{ij}(x)\right)=0\iff r_i(x)=r_j(x).$$
Then \eqref{eq:Z-equiv} is equivalent to 
\be\label{eq:Z-equiv2}x_ix_j(r_i(x)-r_j(x))=0.\ee
To conclude the proof, we are simply left with showing that a configuration $x\in\mc X$ satisfies \eqref{eq:Z-equiv2} if and only if  it is critical, i.e., it belongs to $\mc Z$. Indeed, if $x\in\mc N_{\mc S}$ for some nonempty subset of actions $\mc S\subseteq\mc A$, then necessarily $r_i(x)=r_j(x)$ for every $i,j\in\mc A$ such that $x_ix_j>0$. On the other hand,  for any $x\in\mc X$ satisfying \eqref{eq:Z-equiv2}, it is immediate to verify that $x\in\mc N_{\mc S}$, where $\mc S=\{i\in\mc A:\,x_i>0\}$ is its support. 
\end{proof}\medskip


\begin{thebibliography}{26}%
\providecommand{\url}[1]{#1}
\csname url@samestyle\endcsname
\providecommand{\newblock}{\relax}
\providecommand{\bibinfo}[2]{#2}
\providecommand{\BIBentrySTDinterwordspacing}{\spaceskip=0pt\relax}
\providecommand{\BIBentryALTinterwordstretchfactor}{4}
\providecommand{\BIBentryALTinterwordspacing}{\spaceskip=\fontdimen2\font plus
\BIBentryALTinterwordstretchfactor\fontdimen3\font minus
  \fontdimen4\font\relax}
\providecommand{\BIBforeignlanguage}[2]{{%
\expandafter\ifx\csname l@#1\endcsname\relax
\typeout{** WARNING: IEEEtran.bst: No hyphenation pattern has been}%
\typeout{** loaded for the language `#1'. Using the pattern for}%
\typeout{** the default language instead.}%
\else
\language=\csname l@#1\endcsname
\fi
#2}}
\providecommand{\BIBdecl}{\relax}
\BIBdecl

\bibitem{Weibull1995}
J.~W. Weibull, \emph{{Evolutionary game theory}}.\hskip 1em plus 0.5em minus
  0.4em\relax MIT Press, 1995.

\bibitem{Bjornerstedt1996}
J.~Bj{\"{o}}rnerstedt and J.~W. Weibull, ``{Nash equilibrium and evolution by
  imitation},'' in \emph{The Rational Foundations of Economic Behavior}, 1996,
  pp. 155--171.

\bibitem{Hofbauer2003}
J.~Hofbauer and K.~Sigmund, ``{Evolutionary game dynamics},'' \emph{Bulletin
  (New Series) of the American Mathematical Society}, vol.~40, no.~4, pp.
  479--519, 2003.

\bibitem{Nachbar1990}
J.~H. Nachbar, ``{``Evolutionary'' selection dynamics in games: Convergence and
  limit properties},'' \emph{International Journal of Game Theory}, vol.~19,
  no.~1, pp. 59--89, mar 1990.

\bibitem{Hofbauer2000}
J.~Hofbauer, ``{From Nash and Brown to Maynard Smith: Equilibria, Dynamics and
  ESS},'' \emph{Selection}, vol.~1, no.~1, pp. 81--88, 2000.

\bibitem{Sandholm2001}
W.~H. Sandholm, ``{Potential Games with Continuous Player Sets},''
  \emph{Journal of Economic Theory}, vol.~97, no.~1, pp. 81--108, mar 2001.

\bibitem{Sandholm2010}
------, \emph{{Population Games and Evolutionary Dynamics}}.\hskip 1em plus
  0.5em minus 0.4em\relax Cambridge University Press, 2010, pp. 153--164,
  221--275.

\bibitem{Bomze2002}
I.~M. Bomze, ``{Regularity versus Degeneracy in Dynamics, Games, and
  Optimization: A Unified Approach to Different Aspects},'' \emph{SIAM Review},
  vol.~44, no.~3, pp. 394--414, jan 2002.

\bibitem{Shamma2005}
J.~S. Shamma and G.~Arslan, ``Dynamic fictitious play, dynamic gradient play,
  and distributed convergence to {N}ash equilibria,'' \emph{IEEE Transactions
  on Automatic Control}, vol.~50, no.~3, pp. 312--327, 2005.

\bibitem{Fox2012}
M.~J. Fox and J.~S. Shamma, ``{Population games, stable games, and
  passivity},'' in \emph{Proceedings of the IEEE Conference on Decision and
  Control}.\hskip 1em plus 0.5em minus 0.4em\relax IEEE, dec 2012, pp.
  7445--7450.

\bibitem{Cressman2014}
R.~Cressman and Y.~Tao, ``The replicator equation and other game dynamics,''
  \emph{Proceedings of the National Academy of Sciences of the United States of
  America}, vol. 111, pp. 10\,810--7, 2014.

\bibitem{Barreiro-Gomez2016}
J.~Barreiro-Gomez, G.~Obando, and N.~Quijano, ``{Distributed Population
  Dynamics: Optimization and Control Applications},'' \emph{IEEE Transactions
  on Systems, Man, and Cybernetics: Systems}, vol.~47, no.~2, pp. 304 -- 314,
  2016.

\bibitem{Hofbauer1998}
J.~Hofbauer and K.~Sigmund, \emph{{Evolutionary games and population
  dynamics}}.\hskip 1em plus 0.5em minus 0.4em\relax Cambridge University
  Press, 1998.

\bibitem{Cooper1999}
R.~W. Cooper, \emph{{Coordination Games. Complementarity and
  Macroeconomics}}.\hskip 1em plus 0.5em minus 0.4em\relax Cambridge University
  Press, 1999.

\bibitem{Easley2010}
D.~Easley and J.~Kleinberg, \emph{{Networks, crowds, and markets: reasoning
  about a highly connected world}}.\hskip 1em plus 0.5em minus 0.4em\relax
  Cambridge University Press, 2010.

\bibitem{Skyrms2004}
B.~Skyrms, ``{The Stag Hunt and the Evolution of Social Structure},''
  \emph{Cambridge University Press}, vol.~1, pp. 1--147, 2004.

\bibitem{Rapoport1966}
A.~Rapoport and A.~M. Chammah, ``{The Game of Chicken},'' \emph{American
  Behavioral Scientist}, vol.~10, no.~3, pp. 10--28, nov 1966.

\bibitem{Sugden1986}
R.~Sugden, \emph{{The Economics of Rights, Co-operation and Welfare}}.\hskip
  1em plus 0.5em minus 0.4em\relax London: Palgrave Macmillan UK, 1986,
  vol.~97.

\bibitem{Kurtz1981}
T.~G. Kurtz, \emph{Approximation of Population Processes}.\hskip 1em plus 0.5em
  minus 0.4em\relax Philadelphia, PA: SIAM, 1981, vol.~36.

\bibitem{Taylor1978}
P.~D. Taylor and L.~B. Jonker, ``{Evolutionarily stable strategies and game
  dynamics},'' \emph{Mathematical Biosciences}, vol.~40, no. 1-2, pp. 145--156,
  jul 1978.

\bibitem{Schuster1983}
P.~Schuster and K.~Sigmund, ``{Replicator dynamics},'' \emph{Journal of
  Theoretical Biology}, vol. 100, no.~3, pp. 533--538, feb 1983.

\bibitem{Benaim2003}
M.~Benaim and J.~W. Weibull, ``{Deterministic Approximation of Stochastic
  Evolution in Games},'' \emph{Econometrica}, vol.~71, no.~3, pp. 873--903, may
  2003.

\bibitem{Monderer1996}
D.~Monderer and L.~S. Shapley, ``{Potential Games},'' \emph{Games and Economic
  Behavior}, vol.~14, no.~1, pp. 124--143, may 1996.

\bibitem{Rosenthal1973}
R.~W. Rosenthal, ``{A class of games possessing pure-strategy Nash
  equilibria},'' \emph{International Journal of Game Theory}, vol.~2, no.~1,
  pp. 65--67, dec 1973.

\bibitem{Marden2012}
J.~R. Marden and J.~S. Shamma, ``{Revisiting log-linear learning: Asynchrony,
  completeness and payoff-based implementation},'' \emph{Games and Economic
  Behavior}, vol.~75, no.~2, pp. 788--808, jul 2012.

\bibitem{Como2013}
G.~Como, K.~Savla, D.~Acemoglu, M.~A. Dahleh, and E.~Frazzoli, ``Stability
  analysis of transportation networks with multiscale driver decisions,''
  \emph{SIAM Journal on Control and Optimization}, vol.~51, no.~1, pp.
  230--252, 2013.
\end{thebibliography}
\end{document}